\UseRawInputEncoding
\documentclass{article}
\usepackage{amsthm}
\usepackage{mathrsfs} 
\usepackage{amsmath,amssymb}
\usepackage{graphicx}
\usepackage{float}
\usepackage{epstopdf}
\usepackage{subfigure}
\usepackage{caption}
\usepackage{multirow}
\usepackage{makecell}
\usepackage[figuresright]{rotating}
\usepackage{booktabs}
\numberwithin{equation}{section}
\newtheorem{theorem}{Theorem}[section]  
\newtheorem{definition}[theorem]{Definition}

\newtheorem{lemma}[theorem]{Lemma}
\newtheorem{example}[theorem]{Example} 
\newtheorem{corollary}[theorem]{Corollary}
\newtheorem{remark}[theorem]{Remark}

\newcommand{\rank}{\mathrm{rank}}

\begin{document}
\title{On twisted generalized Reed-Solomon codes with $\ell$ twists}
\author{Haojie Gu\footnote{Haojie Gu is with the School of Mathematical Sciences, Capital Normal University, Beijing, China, 100048. Email: 2200502051@cnu.edu.cn.}
\and Jun Zhang\footnote{Jun Zhang is with the School of Mathematical Sciences, Capital Normal University, Beijing, China, 100048. Email: junz@cnu.edu.cn.}
}

\date{}
\maketitle

\begin{abstract}
	
In this paper, we study a class of twisted generalized Reed-Solomon (TGRS) codes with general $\ell$ twists. A sufficient and necessary condition for the TGRS codes to be MDS or $\ell$-MDS ($\ell<\min\{k,n-k\}$) is determined. A sufficient and necessary condition that such a TGRS code is self-dual for $\ell\leq \lfloor\frac{k-1}{3}\rfloor$ is also presented. Finally, we give an explicit construction of self-dual TGRS codes. And examples of self-dual MDS TGRS codes for small $\ell$ are given.

\begin{flushleft}
	\textbf{Keywords:} Twisted generalized Reed-Solomon codes, Self-dual codes, MDS codes
\end{flushleft}
\end{abstract}

\section{Introduction}
Let $q$ be a power of the prime $p$, $\mathbb{F}_{q}$ be the $q$ elements finite field and $\mathbb{F}_{q}^{*}=\mathbb{F}_{q}\backslash\{0\}$. A linear code $\mathcal{C}\subseteq \mathbb{F}_{q}^n$ with dimension $k$ and minimum distance $d$ will be called a $[n,k,d]_{q}$ linear code. The well-known Singleton bound says that $d\leq n-k+1$ for any code $\mathcal{C}=[n,k,d]_{q}$. The non-negative integer $S(\mathcal{C})=n-k+1-d$ is called the Singleton defect of the code $\mathcal{C}$\cite{de1996almost}. If $S(\mathcal{C})=0$, then $\mathcal{C}$ is called a maximum distance separable (MDS) code. If $S(\mathcal{C})=1$, then $\mathcal{C}$ is called an almost-MDS (AMDS) code. If $S(\mathcal{C})=S(\mathcal{C}^\perp)=1$, then $\mathcal{C}$ is called a near-MDS (NMDS) code. More generally, if $S(\mathcal{C})=S(\mathcal{C}^{\perp})=m$, then $\mathcal{C}$ is called $m$-MDS. Generalized Reed-Solomon(GRS) codes are the most important MDS codes family as they can correct burst and provide high fidelity in CD players. In recent years, constructions of self-dual MDS codes via GRS codes become a hot topic~\cite{fang2019new,FZXF22,jin2018self,jin2016new,lebed2019construction,liu2021construction,zhang2020unified}.

The TGRS codes are generalizations of GRS codes and they were firstly introduced in \cite{beelen2017twisted}. Unlike RS codes, TGRS codes may not be MDS codes. The authors characterized the condition that a TGRS code is MDS and gave two explicit constructions in the paper~\cite{beelen2017twisted}. Afterwards, the properties of TGRS codes
and constructions of self-dual TGRS codes are studied extensively~\cite{betsumiya2003self,huang2021mds,lavauzelle2020cryptanalysis,liu2021construction,9828478,zhang2022class,zhu2021self}. In \cite{huang2021mds}, Huang et al. gave a sufficient and necessary condition that a TGRS code with a single twist is self-dual, and constructed some MDS or NMDS self-dual TGRS codes. In \cite{zhang2022class}, Zhang et al. studied the properties of a class of TGRS codes, such as self-dualness, NMDS or MDS property and so on.  In \cite{9828478}, Sui et al. determined a sufficient and necessary condition that a TGRS code with two twists is MDS. Then they gave a sufficient and necessary condition that a TGRS code with two twists is self-dual, and constructed some MDS, NMDS or $2$-MDS self-dual TGRS codes with two twists. In this paper, we generalize the results for general $\ell$ twists.

 This paper is organized as follows. In Section 2, we show some basic notations and results about TGRS codes. In Section 3, we determine a sufficient and necessary condition that a TGRS code with $\ell$ twists is MDS.   In Section 4,  we characterize the dual codes of TGRS codes and determine a sufficient and necessary condition that a TGRS code with $\ell$ twists is $\ell$-MDS for $\ell<\min\{k,n-k\}$. In section 5, we give a sufficient and necessary condition on self-dual TGRS codes with $\ell$ twists for $\ell\leq \lfloor\frac{k-1}{3}\rfloor$. Finally, we give an explicit construction of self-dual TGRS codes. Also, examples of self-dual MDS TGRS codes for small $\ell$ are given. In Section 6, we conclude our work.

\section{Preliminaries}
 Given a vector
$
\boldsymbol{\alpha}=\left(\alpha_{1}, \alpha_{2}, \ldots, \alpha_{n}\right) \in \mathbb{F}_{q}^{n},
$
where $\alpha_{1}, \alpha_{2}, \ldots, \alpha_{n}$ are distinct elements in $\mathbb{F}_{q}$, usually, $\alpha_{1}, \alpha_{2}, \ldots, \alpha_{n}$ are called evaluation points. Next, given another vector
$
\boldsymbol{v}=\left(v_{1}, v_{2}, \ldots, v_{n}\right) \in\left(\mathbb{F}_{q}^{*}\right)^{n},
$
the evaluation map associated with $\boldsymbol{\alpha}$ and $\boldsymbol{v}$ is defined as
$$
e v_{\boldsymbol{\alpha}, \boldsymbol{v}}: \mathbb{F}_{q}[x]  \mapsto \mathbb{F}_{q}^{n},f(x)\mapsto ev_{\boldsymbol{\alpha},\boldsymbol{v}}(f):=(v_{1}f(\alpha_{1}),v_{2}f(\alpha_{2}),\cdots,v_{n}f(\alpha_{n})).
$$
 In this sense, an $[n, k]$ generalized Reed-Solomon code $G R S_{k}(\boldsymbol{\alpha}, \boldsymbol{v})$ associated with $\boldsymbol{\alpha}$ and $\boldsymbol{v}$ is defined as
$$ GRS_{k}(\boldsymbol{\alpha},\boldsymbol{v}):=\left\{ev_{\boldsymbol{\alpha},\boldsymbol{v}}(f(x)):f(x)\in \mathbb{F}_{q}[x]_{k}\right\},
$$
where $\mathbb{F}_{q}[x]_{k}:=\left\{f(x)\in \mathbb{F}_{q}[x]:\deg(f(x))<k\right\}$. After adding some monomials (called twists) into different positions (called hooks) of each $f(x)$ in $\mathbb{F}_{q}[x]_{k}$, the GRS code can be generalized as follows:

\begin{definition}[\cite{beelen2018structural}]\label{defcont:1}
	For two positive integers $l, k$ and $l \leq k \leq n \leq q$, suppose that $h=\left(h_{1}, h_{2}, \ldots, h_{l}\right)$, where $0 \leq h_{i} \leq k-1$ are distinct, $t=\left(t_{1}, t_{2}, \ldots, t_{l}\right)$, where $0 \leq t_{i}<n-k$ are also distinct, and $\eta=\left(\eta_{1}, \eta_{2}, \ldots, \eta_{l}\right) \in \mathbb{F}_{q}^l$. Then
	\begin{equation*}
	\mathcal{S}=\left\{\sum_{i=0}^{k-1} f_{i} x^{i}+\sum_{j=1}^{l} \eta_{j} f_{h_{j}} x^{k+t_{j}}\,:\, \mbox{for all } f_{i} \in \mathbb{F}_{q}, 0 \leq i \leq k-1\right\}
	\end{equation*}
	is a $k$-dimensional subspace of  $\mathbb{F}_{q}[x]$ over $\mathbb{F}_{q}$. Furthermore, let $\boldsymbol{\alpha}=\left(\alpha_{1}, \alpha_{2}, \ldots,\right.$
	   $\alpha_{n}) \in \mathbb{F}_{q}^{n}$, where $\alpha_{i}, i=1,2, \ldots, n$ are distinct and $\boldsymbol{v}=\left(v_{1}, v_{2}, \ldots, v_{n}\right) \in\left(\mathbb{F}_{q}^{*}\right)^{n}$. The linear code
	\begin{equation*}
	\mathcal{C}=\left\{e v_{\boldsymbol{\alpha}, \boldsymbol{v}}(f(x))\,:\, f(x) \in \mathcal{S}\right\}
	\end{equation*}
	is called a twisted generalized Reed-Solomon (TGRS) code. 
\end{definition} 
In this paper, we shall
consider the case $\ell< \min\{k,n-k\}, {h}=(k-1, k-2,\cdots,k-\ell), {t}=(0,1,\cdots,\ell-1)$ and $\eta=\left(\eta_{1}, \eta_{2},\cdots,\eta_{\ell}\right) \in\left(\mathbb{F}_{q}^{*}\right)^{\ell}$, unless otherwise specified. Let
\begin{equation}\label{equ:2.1}
\mathcal{S}=\left\{\sum_{i=0}^{k-1} f_{i} x^{i}+\sum\limits_{i=0}^{\ell-1}\eta_{i+1}f_{k-\ell+i}x^{k+i}: \mbox { for all } f_{i} \in \mathbb{F}_{q},\ 0 \leq i \leq k-1\right\},
\end{equation}
$\boldsymbol{\alpha}=\left(\alpha_{1}, \alpha_{2}, \ldots, \alpha_{n}\right)$ and $\boldsymbol{v}=\left(v_{1}, v_{2}, \ldots, v_{n}\right)$, where $\alpha_{1}, \alpha_{2}, \ldots, \alpha_{n}$ are distinct elements of $\mathbb{F}_{q}$ and $v_{1}, v_{2}, \ldots, v_{n} \in \mathbb{F}_{q}^{*}$. Then we will focus on the following TGRS code:
\begin{equation}\label{equ:2.2}
\mathcal{C}=\left\{e v_{\boldsymbol{\alpha}, \boldsymbol{v}}(f(x))\,:\, f(x) \in \mathcal{S}\right\} .
\end{equation}

\section{On minimum distances of TGRS codes $\mathcal{C}$}

  In this section, we study the minimums distances of TGRS codes $\mathcal{C}$. Up to the equivalence of code, we may assume that $\boldsymbol{v}=\boldsymbol{1}$, i.e.,  $\mathcal{C}=e v_{\boldsymbol{\alpha},\boldsymbol{1}}(\mathcal{S})$. 
  Obviously, the code $\mathcal{C}$ has generator matrix
\begin{equation}\label{equ:3.1}
G=\begin{pmatrix}
1&\cdots&1\\
\alpha_{1}&\cdots&\alpha_{n}\\
\vdots&\vdots&\vdots\\
\alpha_{1}^{k-\ell-1}&\cdots&\alpha_{n}^{k-\ell-1}\\
\alpha_{1}^{k-\ell}+\eta_{1}\alpha_{1}^k&\cdots&\alpha_{n}^{k-\ell}+\eta_{1}\alpha_{n}^k\\
\alpha_{1}^{k-\ell+1}+\eta_{2}\alpha_{1}^{k+1}&\cdots&\alpha_{n}^{k-\ell+1}+\eta_{2}\alpha_{n}^{k+1}\\
\vdots&\vdots&\vdots\\
\alpha_{1}^{k-1}+\eta_{\ell}\alpha_{1}^{k+\ell-1}&\cdots&\alpha_{n}^{k-1}+\eta_{\ell}\alpha_{n}^{k+\ell-1}\\
\end{pmatrix}.
\end{equation}
Since the TGRS code $\mathcal{C}$ has a sub-code of the GRS code $GRS_{k+\ell}(\boldsymbol{\alpha},\boldsymbol{1})$, the minimum distance $d(\mathcal{C})\geq n-k-\ell +1.$ Together with the Singleton bound, we have 
\[
n-k-\ell +1\leq d(\mathcal{C})\leq n-k+1.
\]
In this section, we will determine  three cases:  $d(\mathcal{C})=n-k+1, d(\mathcal{C})=n-k$ or $d(\mathcal{C})=n-k-\ell+1$. 

The following lemma is straightforward but plays an important role in determining the condition for MDS TGRS code $\mathcal{C}$.
\begin{lemma}\label{lemconst:3.1}
	If $A_{t}=\begin{pmatrix}
	c_{0}&0&\cdots&0\\
	c_{1}&c_{0}&\cdots&0\\
	\vdots&\vdots&\vdots&\vdots\\
	c_{t}&c_{t-1}&\cdots&c_{0}
	\end{pmatrix}$ where $c_{0}=1, c_1, c_2, \cdots, c_t\in \mathbb{F}_q$, then $$A_{t}^{-1}=\begin{pmatrix}
	e_{0}&0&\cdots&0\\
	e_{1}&e_{0}&\cdots&0\\
	\vdots&\vdots&\vdots&\vdots\\
	e_{t}&e_{t-1}&\cdots&e_{0}
	\end{pmatrix}$$
	where $e_{0}=1$ and $e_{i}=-\sum\limits_{j=0}^{i-1}e_{j}c_{i-j},1\leq i\leq t$.
\end{lemma}
\begin{theorem}\label{theconst:3.2}
	Suppose that $3\leq k<n$ and $\alpha_{1},\alpha_{2}\cdots,\alpha_{n}$ are distinct elements of $\mathbb{F}_{q}$. Then $\mathcal{C}=ev_{\boldsymbol{\alpha},\boldsymbol{1}}(\mathcal{S})$ is MDS if and only if $(\eta_{1},\cdots,\eta_{\ell})\in\Omega$, where
	\begin{equation*}
	\begin{aligned}
	&\Omega=\left\{ (\eta_{1},\cdots,\eta_{\ell})\in\mathbb{F}_{q}^\ell:\mbox{for each k-subset I} \subseteq [n],\prod\limits_{i\in I}(x-\alpha_{i})=\sum\limits_{i=0}^{k}c_{i}x^{k-i},\mbox{let}\right.\\
	&g_{k-j}^{(t)}=-\sum\limits_{i=0}^{\min\{t,k-j\}}e_{t-i}c_{j+i},\mbox{where}\   e_{0},\cdots,e_{t}\ \mbox{are in Lemma~\ref{lemconst:3.1}}, 0\leq t<\ell,\\
	&1\leq j\leq\ell,\mbox{it holds}\\
	&\left.\left|\begin{array}{ccccc}
	1+\eta_{1}g_{k-\ell}^{(0)}&\eta_{1}g_{k-\ell+1}^{(0)}&\cdots&\eta_{1}g_{k-2}^{(0)}&\eta_{1}g_{k-1}^{(0)}\\
	\eta_{2}g_{k-\ell}^{(1)}&1+\eta_{2}g_{k-\ell+1}^{(1)}&\cdots&\eta_{2}g_{k-2}^{(1)}&\eta_{2}g_{k-1}^{(1)}\\
	\vdots&\vdots&\vdots&\vdots&\vdots\\
	\eta_{\ell}g_{k-\ell}^{(\ell-1)}&\eta_{\ell}g_{k-\ell+1}^{(\ell-1)}&\cdots&\eta_{\ell}g_{k-2}^{(\ell-1)}&1+\eta_{\ell}g_{k-1}^{(\ell-1)}\\
	\end{array}\right| \neq 0 \right\}.
	\end{aligned}
	\end{equation*}
\end{theorem}
\begin{proof}
	For simplicity, we firstly deal with the case $I=\{1,2,\cdots, k\}$. That is, consider the evaluation points $\alpha_{1},\alpha_{2},\cdots,\alpha_{k}$. Let $\prod\limits_{i=1}^{k}(x-\alpha_{i})=\sum\limits_{j=0}^{k}c_{j}x^{k-j}$. For $0\leq t<\ell$, let
	 $$\left(g_{0}^{(t)}, \ldots, g_{k-1}^{(t)}\right)
	=\left(\alpha_{1}^{k+t}, \ldots, \alpha_{k}^{k+t}\right)\cdot
	\left(\begin{array}{cccc}
	1 & 1 & \cdots & 1 \\
	\alpha_{1} & \alpha_{2} & \cdots & \alpha_{k} \\
	\vdots & \vdots & & \vdots \\
	\alpha_{1}^{k-1} & \alpha_{2}^{k-1} & \cdots & \alpha_{k}^{k-1}
	\end{array}\right)^{-1}.$$
	It means that $\sum\limits_{i=0}^{k-1}g_{i}^{(t)}\alpha_{j}^{i}=\alpha_{j}^{k+t},1\leq j\leq k, 0\leq t<\ell$. Therefore, $\alpha_{1}, \alpha_{2}, \ldots, \alpha_{k}$ are roots of the polynomial $f_{t}(x)=x^{k+t}-\sum\limits_{i=0}^{k-1} g_{i}^{(t)}x^{i}$. So there is $h_{t}(x)=\sum\limits_{i=0}^{t}a_{i}^{(t)}x^i$ such that 
	\begin{equation*}
	\left( \sum\limits_{i=0}^{t}a_{i}^{(t)}x^i\right)\cdot\left(\sum\limits_{j=0}^{k}c_{j}x^{k-j}\right)=h_{t}(x)\prod\limits_{i=1}^{k}(x-\alpha_{i})=x^{k+t}-\sum_{i=0}^{k-1} g_{i}^{(t)}x^{i}.
	\end{equation*}
	Comparing the coefficients of the leftmost side and the rightmost side of the above equation, we have \begin{equation*}
	\left\{
	\begin{array}{ll}
	(a_{0}^{(t)},a_{1}^{(t)},\cdots,a_{t}^{(t)})A_{t}=(0,0,\cdots,0,1),&0\leq t<\ell\\
	g_{k-j}^{(t)}=-\sum\limits_{i=0}^{\min\{t,k-j\}}a_{i}^{(t)}c_{i+j},&1\leq j\leq k-1
	\end{array}.\right.
	\end{equation*}
	By Lemma~\ref{lemconst:3.1}, we know  $$(a_{0}^{(t)},a_{1}^{(t)},\cdots,a_{t}^{(t)})=(0,0,\cdots,0,1)A_{t}^{-1}=(e_{t},e_{t-1},\cdots,e_{0})$$ and $$g_{k-j}^{(t)}=-\sum\limits_{i=0}^{\min\{t,k-j\}}e_{t-i}c_{i+j},\quad 0\leq t< \ell, 1\leq j< \ell.$$ 
	So for $I=\{1,2,\cdots,k\}$, we have {\small
	\begin{equation*}
   \begin{aligned}
     & \left|G_{I}\right|=\left|
	\begin{array}{ccc}
	1&\cdots&1\\
	\alpha_{1}&\cdots&\alpha_{k}\\
	\vdots&\vdots&\vdots\\
	\alpha_{1}^{k-\ell-1}&\cdots&\alpha_{k}^{k-\ell-1}\\
	\alpha_{1}^{k-\ell}+\eta_{1}\alpha_{1}^k&\cdots&\alpha_{k}^{k-\ell}+\eta_{1}\alpha_{k}^k\\
	\vdots&\vdots&\vdots\\
	\alpha_{1}^{k-1}+\eta_{\ell}\alpha_{1}^{k+\ell-1}&\cdots&\alpha_{k}^{k-1}+\eta_{\ell}\alpha_{k}^{k+\ell-1}\\
	\end{array}\right|\\
	&\overset{(1)}{=}\left|\begin{array}{ccc}
	1&\cdots&1\\
	\alpha_{1}&\cdots&\alpha_{k}\\
	\vdots&\vdots&\vdots\\
	\alpha_{1}^{k-\ell-1}&\cdots&\alpha_{k}^{k-\ell-1}\\
	(1+\eta_{1}g_{k-\ell}^{(0)})\alpha_{1}^{k-\ell}+\eta_{1}\sum\limits_{i=k-\ell+1}^{k-1}g_{i}^{(0)}\alpha_{1}^{i}&\cdots&(1+\eta_{1}g_{k-\ell}^{(0)})\alpha_{k}^{k-\ell}+\eta_{1}\sum\limits_{i=k-\ell+1}^{k-1}g_{i}^{(0)}\alpha_{k}^{i}\\
	\vdots&\vdots&\vdots\\
	(1+\eta_{\ell}g_{k-1}^{(\ell-1)})\alpha_{1}^{k-1}+\eta_{\ell}\sum\limits_{i=k-\ell}^{k-2}g_{i}^{(\ell-1)}\alpha_{1}^{i}&\cdots&(1+\eta_{\ell}g_{k-1}^{(\ell-1)})\alpha_{k}^{k-1}+\eta_{\ell}\sum\limits_{i=k-\ell}^{k-2}g_{i}^{(\ell-1)}\alpha_{k}^{i}
	\end{array}\right|\\
	\end{aligned}
\end{equation*}}
 {\small
	\begin{equation*}
		\begin{aligned}
	&=\left|\begin{array}{ccccc}
	\boldsymbol{I}_{k-\ell}&\boldsymbol{0}&\boldsymbol{0}&\cdots&\boldsymbol{0}\\
     \boldsymbol{0}&1+\eta_{1}g_{k-\ell}^{(0)}&\eta_{1}g_{k-\ell+1}^{(0)}&\cdots&\eta_{1}g_{k-1}^{(0)}\\
	\boldsymbol{0}&\eta_{2}g_{k-\ell}^{(1)}&1+\eta_{2}g_{k-\ell+1}^{(1)}&\cdots&\eta_{2}g_{k-1}^{(1)}\\
	\vdots&\vdots&\vdots&\vdots&\vdots\\
	\boldsymbol{0}&\eta_{\ell}g_{k-\ell}^{(\ell-1)}&\eta_{\ell}g_{k-\ell+1}^{(\ell-1)}&\cdots&1+\eta_{\ell}g_{k-1}^{(\ell-1)}\\
	\end{array}\right|\cdot\left| \begin{array}{cccc}
	1&1&\cdots&1\\
	\alpha_{1}&\alpha_{2}&\cdots&\alpha_{k}\\
	\alpha_{1}^2&\alpha_{2}^2&\cdots&\alpha_{k}^2\\
	\vdots&\vdots&\vdots&\vdots\\
	\alpha_{1}^{k-1}&\alpha_{2}^{k-1}&\cdots&\alpha_{k}^{k-1}
	\end{array}\right|\\
	&=\left|\begin{array}{cccc}
	1+\eta_{1}g_{k-\ell}^{(0)}&\eta_{1}g_{k-\ell+1}^{(0)}&\cdots&\eta_{1}g_{k-1}^{(0)}\\
	\eta_{2}g_{k-\ell}^{(1)}&1+\eta_{2}g_{k-\ell+1}^{(1)}&\cdots&\eta_{2}g_{k-1}^{(1)}\\
	\vdots&\vdots&\vdots&\vdots\\
	\eta_{\ell}g_{k-\ell}^{(\ell-1)}&\eta_{\ell}g_{k-\ell+1}^{(\ell-1)}&\cdots&1+\eta_{\ell}g_{k-1}^{(\ell-1)}\\
	\end{array}\right|\cdot\prod\limits_{1\leq j<i\leq k}(\alpha_{i}-\alpha_{j}).
	\end{aligned}
	\end{equation*}}
	The equality $(1)$ follows from: substitutions $\alpha_j^{k+t}=\sum\limits_{i=0}^{k-1}g_{i}^{(t)}\alpha_{j}^{i},\,j=1,2,\cdots,k, t=0,1,\cdots, \ell-1$; and the terms consisting of $\alpha_j^i \,(j=1,2,\cdots,k, i=0,1,\cdots, k-\ell-1)$ in the $k-\ell+1, k-\ell+2,\cdots, k$-th rows are absorbed by the first $k-\ell$ rows. 
	
	Thus, according to the generality of $I$, $\mathcal{C}=ev_{\boldsymbol{\alpha},\boldsymbol{1}}(\mathcal{S})$ is MDS if and only if all $k\times k$ minors of G are non-zero, if and only if $(\eta_{1},\cdots,\eta_{\ell})\in\Omega$.
\end{proof}

 For $\ell=2,3$, we obtain the following sufficient and necessary condition about $\mathcal{C}=ev_{\boldsymbol{\alpha},\boldsymbol{1}}(\mathcal{S})$ to be MDS:
\begin{corollary}\label{corcont:3.3}
	Suppose that $3\leq k<n,\ell=2$ and $\alpha_{1},\alpha_{2}\cdots,\alpha_{n}$ are distinct elements of $\mathbb{F}_{q}$. Then $\mathcal{C}=ev_{\boldsymbol{\alpha},\boldsymbol{1}}(\mathcal{S})$ is MDS if and only if for each k-subset $I\subseteq [n],\prod\limits_{i\in I}(x-\alpha_{i})=\sum\limits_{i=0}^{k}c_{i}x^{k-i}$, it holds $1+\eta_{2}(c_{1}^2-c_{2})-\eta_{1}c_{2}+\eta_{1}\eta_{2}(c_{2}^2-c_{1}c_{3})\neq 0$.
\end{corollary}

\begin{remark}\label{remcont:3.4}
	For $\ell=2$, the twists of the above corollary is different from \cite{9828478}, so we obtain a different necessary and sufficient condition comparing with the result of~\cite{9828478}.
\end{remark}

\begin{corollary}\label{corcont:3.4}
	Suppose that $5\leq k<n,\ell=3$ and $\alpha_{1},\alpha_{2}\cdots,\alpha_{n}$ are distinct elements of $\mathbb{F}_{q}$. Then $\mathcal{C}=ev_{\boldsymbol{\alpha},\boldsymbol{1}}(\mathcal{S})$ is MDS if and only if for each k-subset $I\subseteq [n],\prod\limits_{i\in I}(x-\alpha_{i})=\sum\limits_{i=0}^{k}c_{i}x^{k-i}$,  let
	$g_{k-j}^{(t)}=-\sum\limits_{i=0}^{t}e_{t-i}c_{j+i}$, where   $e_{0}=1,e_{1}=-c_{1},e_{2}=-c_{2}+c_{1}^2, 0\leq t\leq 2,1\leq j\leq 3$, it holds\begin{equation*}
	\left|
	\begin{array}{ccc}
	1+\eta_{1}g_{k-3}^{(0)}&\eta_{1}g_{k-2}^{(0)}&\eta_{1}g_{k-1}^{(0)}\\
	\eta_{2}g_{k-3}^{(1)}&1+\eta_{2}g_{k-2}^{(1)}&\eta_{2}g_{k-1}^{(1)}\\
	\eta_{3}g_{k-3}^{(2)}&\eta_{3}g_{k-2}^{(2)}&1+\eta_{3}g_{k-1}^{(2)}
	\end{array}
	\right|\neq 0
\end{equation*}
\end{corollary}
	
	Next, we will give two explicit examples based on the above two corollaries.

\begin{example}\label{exacont:3.6}
	Let $q=11,\ell=2,$ and $\boldsymbol{\alpha}=(1,2,3,5,6,8,9,10)$. By Corollary~\ref{corcont:3.3}, $\mathcal{C}=ev_{\boldsymbol{\alpha},\boldsymbol{1}}(\mathcal{S})$ is an $[8,k,9-k]_q$ MDS code if and only if for each $k$-subset $I\subseteq [n],\prod\limits_{i\in I}(x-\alpha_{i})=\sum\limits_{i=0}^{k}c_{i}x^{k-i}$, it holds $1+\eta_{2}(c_{1}^2-c_{2})-\eta_{1}c_{2}+\eta_{1}\eta_{2}(c_{2}^2-c_{1}c_{3})\neq 0$, where $c_{1}=-\sum\limits_{i\in I}\alpha_{i},c_{2}=\sum\limits_{i,j\in I\atop i<j}\alpha_{i}\alpha_{j}$ and $c_{3}=-\sum\limits_{i,j,t\in I\atop i<j<t}\alpha_{i}\alpha_{j}\alpha_{t}$.
	By {MATLAB}, we get the following table of pairs $(\eta_{1},\eta_{2})$ of parameters for dimensions $k\in \{3,4,5\}$. 

\begin{table}[H]
	\centering
	\begin{tabular}{|c|c|c|}
		\hline
		\multirow{2}[0]{*}{dimension $k$} & \multirow{2}[0]{*}{parameter $(\eta_{1},\eta_{2})$} & {\multirow{2}[0]{*}{MDS code}} \\
		& & \\
		\hline
		\multirow{2}[0]{*}{$k=3$}&$(0,0)$  & \multirow{2}[0]{*}{$[8,3,6]$}\\
		&$(2,9)$  & \\
		\hline
		\multirow{3}[0]{*}{$k=4$}&$(0,0)$  & \multirow{3}[0]{*}{$[8,4,5]$}\\
		&$(4,4)$  & \\
		&$(6,6)$  & \\
		\hline
		\multirow{2}[0]{*}{$k=5$}&$(0,0)$  & \multirow{2}[0]{*}{$[8,5,4]$}\\
		&$(9,10)$  & \\
		\hline
	\end{tabular}%
\end{table}
Moreover, there are $14$ pairs $(\eta_{1},\eta_{2})$ of parameters to make $\mathcal{C}=ev_{\boldsymbol{\alpha},\boldsymbol{1}}(\mathcal{S})$ to be MDS codes $[8,6,3]$.  And there $70$ parameters $(\eta_{1},\eta_{2})$ to make $\mathcal{C}=ev_{\boldsymbol{\alpha},\boldsymbol{1}}(\mathcal{S})$ to be MDS codes $[8,7,2]$. 
	
\end{example}

\begin{example}\label{exacont:3.7}
	Let $q=13,\ell=3$ and $\boldsymbol{\alpha}=\left(0,1,2,3,4,5,6,9,10,12\right)$. By Corollary~\ref{corcont:3.4}, $\mathcal{C}=ev_{\boldsymbol{\alpha},\boldsymbol{1}}(\mathcal{S})$ is a $[10,k,11-k]_q$ MDS code if and only if for each $k$-subset $I\subseteq [n],\prod\limits_{i\in I}(x-\alpha_{i})=\sum\limits_{i=0}^{k}c_{i}x^{k-i}$,  let
	$g_{k-j}^{(t)}=-\sum\limits_{i=0}^{t}e_{t-i}c_{j+i}$, where   $e_{0}=1,e_{1}=-c_{1},e_{2}=-c_{2}+c_{1}^2, 0\leq t\leq 2,1\leq j\leq 3$, it holds\begin{equation*}
	\left|
	\begin{array}{ccc}
	1+\eta_{1}g_{k-3}^{(0)}&\eta_{1}g_{k-2}^{(0)}&\eta_{1}g_{k-1}^{(0)}\\
	\eta_{2}g_{k-3}^{(1)}&1+\eta_{2}g_{k-2}^{(1)}&\eta_{2}g_{k-1}^{(1)}\\
	\eta_{3}g_{k-3}^{(2)}&\eta_{3}g_{k-2}^{(2)}&1+\eta_{3}g_{k-1}^{(2)}
	\end{array}
	\right|\neq 0
	\end{equation*}
	By {MATLAB}, we get the following table about parameters $(\eta_{1},\eta_{2},\eta_{3})$ and MDS codes with dimension $k\,(5\leq k\leq 9)$. 
	
	\begin{table}[H]
		\centering
		\begin{tabular}{|c|c|c|}
			\hline
			\multirow{2}[0]{*}{dimension $k$} & \multirow{2}[0]{*}{the number of parameter $(\eta_{1},\eta_{2},\eta_{3})$} & {\multirow{2}[0]{*}{MDS code}} \\
			& & \\
			\hline
			\multirow{2}[0]{*}{$k=5$}& \multirow{2}[0]{*}{$197$} & \multirow{2}[0]{*}{$[10,5,6]$}\\
			&&\\
			\hline
			\multirow{2}[0]{*}{$k=6$}&\multirow{2}[0]{*}{$234$}  & \multirow{2}[0]{*}{$[10,6,5]$}\\
			&&\\
			\hline
			\multirow{2}[0]{*}{$k=7$}& \multirow{2}[0]{*}{$500$} & \multirow{2}[0]{*}{$[10,7,4]$}\\
			&&\\
			\hline
			\multirow{2}[0]{*}{$k=8$}&\multirow{2}[0]{*}{$1216$}  & \multirow{2}[0]{*}{$[10,8,3]$}\\
			&&\\
			\hline
			\multirow{2}[0]{*}{$k=9$}&\multirow{2}[0]{*}{$1619$}  & \multirow{2}[0]{*}{$[10,9,2]$}\\
			&&\\
			\hline
		\end{tabular}%
	\end{table}
	For example, for $k=5$ and $(\eta_{1}, \eta_{2}, \eta_{3})=(2,3,6)$, let 
	\begin{equation*}
	\mathcal{S}=\left\{\sum\limits_{i=0}^{4}f_{i}x^i+2f_{2}x^5+3f_{3}x^6+6f_{4}x^7:\mbox{for all}\  f_{i}\in\mathbb{F}_{13},0\leq i\leq 4\right\}.
	\end{equation*}
	The TGRS $\mathcal{C}=ev_{\boldsymbol{\alpha},\boldsymbol{1}}(\mathcal{S})$ is an MDS code $[10,5,6]$.
	We will see later in Section~5 that there exists $\boldsymbol{v}\in\mathbb{F}_{13^2}^{10}\setminus\{\boldsymbol{0}\}$ such that $\mathcal{C}=ev_{\boldsymbol{\alpha},\boldsymbol{v}}(\mathcal{S})$ is a  self-dual MDS code $[10,5,6]$ over the extension field $\mathbb{F}_{13^2}$, where 
	 \begin{equation*}
	 \mathcal{S}=\left\{\sum\limits_{i=0}^{4}f_{i}x^i+2f_{2}x^5+3f_{3}x^6+6f_{4}x^7:\mbox{for all}\  f_{i}\in\mathbb{F}_{13^2},0\leq i\leq 4\right\}.
	 \end{equation*}

\end{example}
Finally, we investigate  the sufficient and necessary condition about that the Singleton defect $S(\mathcal{C})=1$ or $\ell$.

\begin{lemma}[\cite{9828478}]\label{lemcont:3.5} 
	An $[n, k]$ linear code $\mathcal{C}$ over $\mathbb{F}_{q}$ satisfies $S(\mathcal{C})=1$ if and only if the generator matrix $G$ of $\mathcal{C}$ satisfies the following conditions:
	
 $(1)$\  There exists $k$ linearly dependent columns in $G$, i.e., $S(\mathcal{C})\neq 0$.
	 
 $(2)$\  Any $k+1$ columns of $G$ be of $\rank$ $k$, i.e., $S(\mathcal{C})\leq 1$.
\end{lemma}

From above result, we can obtain the sufficient and necessary condition of AMDS.

\begin{corollary}
	With the notation as in Theorem~\ref{theconst:3.2}, let $(\eta_{1},\cdots,\eta_{\ell})\in \mathbb{F}_{q}^{\ell}\setminus\Omega$. Then $\mathcal{C}$ is AMDS if and only if for each $(k+1)$-subset $J\subseteq \{1,2,\cdots,n\}$, there is a k-susbet $I\subseteq J$, such that let $\prod\limits_{i\in I}(x-\alpha_{i})=\sum\limits_{i=0}^{k}c_{i}x^{k-i}$ and $g_{k-j}^{(t)}=-\sum\limits_{i=0}^{\min\{t,k-j\}}e_{t-i}c_{j+i}$, where $e_{0},\cdots,e_{t}$ are in Lemma~\ref{lemconst:3.1}, $0\leq t<\ell,1\leq j\leq \ell$, it holds
	\begin{equation*}
	\left|\begin{array}{ccccc}
	1+\eta_{1}g_{k-\ell}^{(0)}&\eta_{1}g_{k-\ell+1}^{(0)}&\cdots&\eta_{1}g_{k-2}^{(0)}&\eta_{1}g_{k-1}^{(0)}\\
	\eta_{2}g_{k-\ell}^{(1)}&1+\eta_{2}g_{k-\ell+1}^{(1)}&\cdots&\eta_{2}g_{k-2}^{(1)}&\eta_{2}g_{k-1}^{(1)}\\
	\vdots&\vdots&\vdots&\vdots&\vdots\\
	\eta_{\ell}g_{k-\ell}^{(\ell-1)}&\eta_{\ell}g_{k-\ell+1}^{(\ell-1)}&\cdots&\eta_{\ell}g_{k-2}^{(\ell-1)}&1+\eta_{\ell}g_{k-1}^{(\ell-1)}\\
	\end{array}\right| \neq 0.    	
	\end{equation*}
\end{corollary}
	
\begin{proof}
	Since $(\eta_{1},\cdots,\eta_{\ell})\notin\Omega,$ by Theorem~\ref{theconst:3.2} we have $S(\mathcal{C})\geq 1$. $\mathcal{C}$ is AMDS if and only if $S(\mathcal{C})\leq 1$, if and only if any $k+1$ columns of G be of $\rank$  $k$ by Lemma~\ref{lemcont:3.5}. That is, for any $(k+1)$-subset $J\subseteq\{1,2,\cdots,n\}$, there is a $k$-subset $I\subseteq J$, such that let $\prod\limits_{i\in I}(x-\alpha_{i})=\sum\limits_{i=0}^{k}c_{i}x^{k-i}$ and $g_{k-j}^{(t)}=-\sum\limits_{i=0}^{\min\{t,k-j\}}e_{t-i}c_{j+i}$, where $e_{0},\cdots,e_{t}$ are in Lemma ~\ref{lemconst:3.1}, $0\leq t<\ell,1\leq j\leq \ell$, it holds 
	\begin{equation*}
	\left|\begin{array}{ccccc}
	1+\eta_{1}g_{k-\ell}^{(0)}&\eta_{1}g_{k-\ell+1}^{(0)}&\cdots&\eta_{1}g_{k-2}^{(0)}&\eta_{1}g_{k-1}^{(0)}\\
	\eta_{2}g_{k-\ell}^{(1)}&1+\eta_{2}g_{k-\ell+1}^{(1)}&\cdots&\eta_{2}g_{k-2}^{(1)}&\eta_{2}g_{k-1}^{(1)}\\
	\vdots&\vdots&\vdots&\vdots&\vdots\\
	\eta_{\ell}g_{k-\ell}^{(\ell-1)}&\eta_{\ell}g_{k-\ell+1}^{(\ell-1)}&\cdots&\eta_{\ell}g_{k-2}^{(\ell-1)}&1+\eta_{\ell}g_{k-1}^{(\ell-1)}\\
	\end{array}\right| \neq 0.    	
	\end{equation*}
\end{proof}

Finally, we determine a condition for $S(\mathcal{C})=\ell$.
\begin{theorem}\label{thmcont:3.7}
	The Singleton defect of $\mathcal{C}=ev_{\boldsymbol{\alpha},\boldsymbol{1}}(\mathcal{S})$ satisfies $S(\mathcal{C})=\ell$ if and only if there exists $k+\ell-1$-subset $I\subseteq \{1,2,\cdots,n\}$ such that $$c_{k+\ell-i}=\eta_{\ell-i+1}c_{k-i},\quad i=1,2,\cdots,\ell,$$ where  $c_0,c_1,\cdots, c_{k+\ell-1}$ satisfy $\prod\limits_{i\in I}(x-\alpha_{i})=\sum\limits_{i=0}^{k+\ell-1}c_{i}x^{k+\ell-1-i}$.
\end{theorem}
\begin{proof}
	We know that $$d(\mathcal{C})=\mathop{\min}\limits_{f\in \mathcal{S}\setminus\{0\}}\#\left\{i\in [n]:f(\alpha_{i})\neq 0\right\}= n-\mathop{\max}\limits_{f\in \mathcal{S}\setminus\{0\}}\#\left\{
	i\in [n]:f(\alpha_{i})=0\right\}.$$
	On one hand, $S(\mathcal{C})=\ell$ if and only if $n-d=k+\ell-1$, if and only if 
	$$\mathop{\max}\limits_{f\in \mathcal{S}\setminus\{0\}}\#\left\{i\in [n]:f(\alpha_{i})=0\right\}=k+\ell-1.$$ On the other hand, any polynomial $ f\in\mathcal{S}\setminus\{0\}$ has degree $\deg(f)\leq k+\ell-1$. Thus, $S(\mathcal{C})=\ell$ if and only if there exists $k+\ell-1$-subset $I\subseteq \{1,2,\cdots,n\}$ such that $\prod\limits_{i\in I}(x-\alpha_{i})\in\mathcal{S}$. That is, there exists $k+\ell-1$-subset $I\subseteq \{1,2,\cdots,n\}$ such that $$c_{k+\ell-i}=\eta_{\ell-i+1}c_{k-i},\quad i=1,2,\cdots,\ell,$$ where  $c_0,c_1,\cdots, c_{k+\ell-1}$ satisfy $\prod\limits_{i\in I}(x-\alpha_{i})=\sum\limits_{i=0}^{k+\ell-1}c_{i}x^{k+\ell-1-i}$.
\end{proof}

\section{ The dual codes of TGRS codes $\mathcal{C}$}
For any two vectors $\boldsymbol{x}=(x_{1},\cdots,x_{n}),\boldsymbol{y}=(y_{1},\cdots,y_{n})\in\mathbb{F}_{q}^n$, the inner product of $\boldsymbol{x}$ and $\boldsymbol{y}$ is defined as $\boldsymbol{x}\cdot\boldsymbol{y}=\sum\limits_{i=1}^{n}x_{i}y_{i}$. The dual code of a linear code $\mathcal{C}$ is defined to be
$$
\mathcal{C}^{\perp}=\left\{\boldsymbol{x} \in \mathbb{F}_{q}^{n}: \boldsymbol{x} \cdot \boldsymbol{c}=\sum_{i=1}^{n} c_{i} x_{i}=0 \mbox { for all } \boldsymbol{c} \in \mathcal{C}\right\}.
$$
In this section, we will devote to determining the dual code or a parity-check matrix of the TGRS codes $\mathcal{C}=ev_{\boldsymbol{\alpha},\boldsymbol{1}}(\mathcal{S})$.

Firstly, we recall a useful result from~\cite{9828478}.
\begin{lemma}[\cite{9828478}]\label{lemcont:4.1}
	Let $\alpha_{1},\cdots,\alpha_{n}$ be distinct elements of $\mathbb{F}_{q}$ and $\prod\limits_{i=1}^{n}(x-\alpha_{i})=\sum\limits_{j=0}^{n}\sigma_{j}x^{n-j}$. Let $\Lambda_{0}=1$ and $\boldsymbol{y}=(\Lambda_{0},\Lambda_{1},\cdots,\Lambda_{n})$ be the unique solution of the following system of equations：
	\begin{equation*}
	\begin{pmatrix}
	\sigma_{0}&0&0&\cdots&0\\
	\sigma_{1}&\sigma_{0}&0&\cdots&0\\
	\sigma_{2}&\sigma_{1}&\sigma_{0}&\cdots&0\\
	\vdots&\vdots&\vdots&\ddots&\vdots\\
	\sigma_{n}&\sigma_{n-1}&\sigma_{n-2}&\cdots&\sigma_{0}
	\end{pmatrix}
	\begin{pmatrix}
	\Lambda_{0}\\\Lambda_{1}\\ \vdots\\ \Lambda_{n}
	\end{pmatrix}=
	\begin{pmatrix}
	1\\0\\ \vdots\\ 0
	\end{pmatrix}.
	\end{equation*}
	For any fixed $0\leq t\leq n$, if $\alpha_{i}^{n-1+t}=\sum\limits_{j=0}^{n-1}f_{j}\alpha_{i}^j$ for $1\leq i\leq n$, then $f_{n-1}=\Lambda_{t}$.
\end{lemma}
\begin{theorem}\label{thecont:4.2}
	Let $\alpha_{1},\alpha_{2},\cdots,\alpha_{n}$ be distinct elements of $\mathbb{F}_{q}$, $\prod\limits_{i=1}^{n}(x-\alpha_{i})=\sum\limits_{j=0}^{n}\sigma_{j}x^{n-j}$ and $u_{i}=\prod\limits_{j=1,j\neq i}^{n}(\alpha_{i}-\alpha_{j})^{-1}$ for $1\leq i\leq n$. If $\prod\limits_{j=1}^{\ell}\eta_{j}\neq 0$, then $\mathcal{C}=ev_{\boldsymbol{\alpha},\boldsymbol{1}}(\mathcal{S})$ has parity check matrix \begin{equation}\label{equ:4.1}
	H=\begin{pmatrix}
	\cdots&u_{j}&\cdots\\
	\cdots&u_{j}\alpha_{j}&\cdots\\
	\vdots&\vdots&\vdots\\
	\cdots&u_{j}\alpha_{j}^{n-k-\ell-1}&\cdots\\
	\cdots&u_{j}\alpha_{j}^{n-k-\ell}\left(1-\eta_{\ell}\sum\limits_{i=0}^{\ell}\sigma_{\ell-i}\alpha_{j}^i\right)&\cdots\\
	\cdots&u_{j}\alpha_{j}^{n-k-\ell}\left(\sum\limits_{i=0}^{1}\sigma_{1-i}\alpha_{j}^i-\eta_{\ell-1}\sum\limits_{i=0}^{\ell+1}\sigma_{\ell+1-i}\alpha_{j}^i\right)&\cdots\\
	\vdots&\vdots&\vdots\\
	\cdots&u_{j}\alpha_{j}^{n-k-\ell}\left(\sum\limits_{i=0}^{\ell-1}\sigma_{\ell-1-i}\alpha_{j}^i-\eta_{1}\sum\limits_{i=0}^{2\ell-1}\sigma_{2\ell-1-i}\alpha_{j}^i\right)&\cdots\\
	\end{pmatrix}_{(n-k)\times n}.
	\end{equation}
\end{theorem}
\begin{proof}
	Firstly, we prove that $\mathrm{rank}(H)=n-k$. Suppose that $(f_{0},\cdots,f_{n-k-1})$ is a solution of the system of equations: $(x_{0},x_{1},\cdots,x_{n-k-1})H=\boldsymbol{0}$. Next, we want to show that $(f_{0},\cdots,f_{n-k-1})=\boldsymbol{0}$. Let $$f(x)=\sum\limits_{i=0}^{n-k-\ell-1}f_{i}x^i+\sum\limits_{j=0}^{\ell-1}f_{n-k-\ell+j}m_{j}(x),$$ where $m_{i}(x)=x^{n-k-\ell}(\sum\limits_{j=0}^i\sigma_{i-j}x^j-\eta_{\ell-i}\cdot\sum\limits_{j=0}^{i+\ell}\sigma_{i+\ell-j}x^j),0\leq i\leq \ell-1$. Then $f(\alpha_{i})=0$ for $1\leq i\leq n$. But the degree of $f(x)$ is $\deg(f(x))\leq n-k+\ell-1<n$. So we have $f(x)=0$. This means that $f_{0}=f_{1}=\cdots=f_{n-k-\ell-1}=0$ and
	\begin{equation*}
	\begin{pmatrix}
	1-\eta_{\ell}\sigma_{\ell}&\sigma_{1}-\eta_{\ell-1}\sigma_{\ell+1}&\cdots&\sigma_{\ell-1}-\sigma_{2\ell-1}\eta_{1}\\
	-\eta_{\ell}\sigma_{\ell-1}&1-\eta_{\ell-1}\sigma_{\ell}&\cdots&\sigma_{\ell-2}-\sigma_{2\ell-2}\eta_{1}\\
	\vdots&\vdots&\vdots&\vdots\\
	-\eta_{\ell}&-\eta_{\ell-1}\sigma_{1}&\cdots&-\sigma_{\ell-1}\eta_{1}\\
	0&-\eta_{\ell-1}&\cdots&-\sigma_{\ell-2}\eta_{1}\\
	\vdots&\vdots&\vdots&\vdots\\
	0&0&\cdots&-\eta_{1}
	\end{pmatrix}\cdot 
	\begin{pmatrix}
	f_{n-k-\ell}\\f_{n-k-\ell+1}\\ \vdots\\ f_{n-k-1}
	\end{pmatrix}=\begin{pmatrix}
	0\\0\\ \vdots\\0
	\end{pmatrix}.
	\end{equation*}
	So $f_{0}=f_{1}=\cdots=f_{n-k-1}=0$. Therefore, the system of $(x_{0},\cdots,x_{n-k-1})H=0$ only has a trivial solution, i.e., $\rank(H)=n-k$.
	
	 Secondly, we prove that $G H^{T}=\boldsymbol{0}$. Let $G^{T}=\left(\boldsymbol{g}_{0}^{T}, \boldsymbol{g}_{1}^{T}, \ldots, \boldsymbol{g}_{k-1}^{T}\right)$ and $H^{T}=$ $\left(\boldsymbol{h}_{0}^{T}, \boldsymbol{h}_{1}^{T}, \ldots, \boldsymbol{h}_{n-k-1}^{T}\right)$, where $\boldsymbol{g}_{i}$ is the $(i+1)$-th row of $G$ and $\boldsymbol{h}_{j}$ is the $(j+1)$-th row of $H$ for all $i=0,1,\cdots,k-1$ and $j=0,1,\cdots,n-k-1$. From the proof of \cite[Theorem 2.4]{huang2021mds}, we know that 
	\begin{equation*}
	\left\{
	\begin{array}{ll}
	\sum\limits_{t=1}^{n}u_{t}\alpha_{t}^i=0,&if\ 0\leq i\leq n-2;\\
	\sum\limits_{t=1}^{n}u_{t}\alpha_{t}^{i}=1,&if\ i=n-1.
	\end{array}\right.
	\end{equation*} 
	By using the above result, it is straightforward to verify that  $\boldsymbol{g}_{i}\boldsymbol{h}_{j}^{T}=\boldsymbol{0},$ for all $ 0\leq i\leq k-\ell-1,0\leq j\leq n-k-1$ and for all $ k-\ell\leq i\leq k-1,0\leq j\leq n-k-\ell-1$.\\
	
	 For $i,j\in\{0,1,\cdots,\ell-1\}$, direct computing shows that  \begin{equation*}
	\begin{aligned}
	&\boldsymbol{g}_{k-\ell+i}\cdot \boldsymbol{h}_{n-k-\ell+j}^{T}\\
	=&\sum\limits_{r=1}^n\left(\alpha_{r}^{k-\ell+i}+\eta_{i+1}\alpha_{r}^{k+i}\right)u_{r}\alpha_{r}^{n-k-\ell}\left(\sum\limits_{w=0}^{j}\sigma_{j-w}\alpha_{r}^w-\eta_{\ell-j}\sum\limits_{w=0}^{\ell+j}\sigma_{\ell+j-w}\alpha_{r}^{w}\right)\\
	=&\sum\limits_{w=0}^{j}\sigma_{j-w}\sum\limits_{r=1}^nu_{r}\alpha_{r}^{n-2\ell+i+w}+\eta_{i+1}\sum\limits_{w=0}^{j}\sigma_{j-w}\sum\limits_{r=1}^{n}u_{r}\alpha_{r}^{n-\ell+i+w}\\
	&-\eta_{\ell-j}\sum\limits_{w=0}^{j+\ell}\sigma_{\ell+j-w}\sum\limits_{r=1}^{n}u_{r}\alpha_{r}^{n-2\ell+i+w}-\eta_{i+1}\eta_{\ell-j}\sum\limits_{w=0}^{\ell+j}\sigma_{\ell+j-w}\sum\limits_{r=1}^{n}u_{r}\alpha_{r}^{n-\ell+i+w}.
	\end{aligned}
	\end{equation*}
	 Next, we prove $\boldsymbol{g}_{k-\ell+i}\cdot\boldsymbol{h}_{n-k-\ell+j}^T=\boldsymbol{0}$ in three cases: $i+j<\ell-1,i+j=\ell-1$ and $i+j>\ell-1$. 
	 
	 If $i+j<\ell-1$, then \begin{equation*}
	\begin{aligned}
	\boldsymbol{g}_{k-\ell+i}\cdot\boldsymbol{h}_{n-k-\ell+j}^T&=-\eta_{i+1}\eta_{\ell-j}\sum\limits_{w=\ell-i-1}^{j+\ell}\sigma_{\ell+j-w}\Lambda_{w-\ell+i+1}\\
	&=-\eta_{i+1}\eta_{\ell-j}\sum\limits_{w=0}^{i+j+1}\sigma_{i+j+1-w}\Lambda_{w}=0
	\end{aligned}
	\end{equation*}
	where the first and last equalities follow from Lemma~\ref{lemcont:4.1}.
	
	 If $i+j=\ell-1$, then \begin{equation*}
	\begin{aligned}
	\boldsymbol{g}_{k-\ell+i}\cdot\boldsymbol{h}_{n-k-\ell+j}^T&=\eta_{i+1}-\eta_{\ell-j}-\eta_{i+1}\eta_{\ell-j}\sum\limits_{w=\ell-i-1}^{j+\ell}\sigma_{\ell+j-w}\Lambda_{w-\ell+i+1}\\
	&=-\eta_{i+1}\eta_{\ell-j}\sum\limits_{w=0}^{i+j+1}\sigma_{i+j+1-w}\Lambda_{w}=0
	\end{aligned}
	\end{equation*}
	where the first and last equalities follow from Lemma~\ref{lemcont:4.1}. 
	
	 If $i+j>\ell-1$, then \begin{equation*}
	\begin{aligned}
	&\boldsymbol{g}_{k-\ell+i}\cdot\boldsymbol{h}_{n-k-\ell+j}^T\\
	=&\eta_{i+1}\sum\limits_{w=\ell-i-1}^{j}\sigma_{j-w}\Lambda_{w-\ell+i+1}-\eta_{\ell-j}\sum\limits_{w=2\ell-i-1}^{j+\ell}\sigma_{\ell+j-w}\Lambda_{w-2\ell+i+1}\\
	&-\eta_{i+1}\eta_{\ell-j}\sum\limits_{w=\ell-i-1}^{j+\ell}\sigma_{\ell+j-w}\Lambda_{w-\ell+i+1}\\
	=&\eta_{i+1}\sum\limits_{w=0}^{i+j+1-\ell}\sigma_{i+j+1-\ell-w}\Lambda_{w}-\eta_{\ell-j}\sum\limits_{w=0}^{i+j+1-\ell}\sigma_{i+j+1-\ell-w}\Lambda_{w}\\
	&-\eta_{i+1}\eta_{\ell-j}\sum\limits_{w=0}^{i+j+1}\sigma_{i+j+1-w}\Lambda_{w}\\
	=&0
	\end{aligned}
	\end{equation*}
	where the first and last equalities follow from Lemma~\ref{lemcont:4.1}.
\end{proof}

Now by applying Theorem~\ref{thmcont:3.7} to the dual codes $\mathcal{C}^{\perp}$, we obtain the folllowing necessary and sufficient condition for $\mathcal{C}$ to be $\ell$-MDS.
\begin{corollary}\label{corcont:4.3}
	Let $G$ in $\left(\ref{equ:3.1}\right)$ and $H$ in $(\ref{equ:4.1})$ be the generator matrix and parity check matrix of $\mathcal{C}$, respectively. Then $\mathcal{C}$ is $\ell$-MDS if and only if the following conditions hold:
	
	 $(1)$\ There exists $k+\ell-1$-subset $I\subseteq \{1,2,\cdots,n\}$ such that $$c_{k+\ell-i}=\eta_{\ell-i+1}c_{k-i},\, i=1,2,\cdots,\ell,$$
	  where $c_0,c_1,\cdots, c_{k+\ell-1}$ satisfy $\prod\limits_{i\in I}(x-\alpha_{i})=\sum\limits_{i=0}^{k+\ell-1}c_{i}x^{k+\ell-1-i}$.
	
	 $(2)$\ There exists $n-k+\ell-1$-subset $J\subseteq \{1,2,\cdots,n\}$ such that the following system of equations has solutions:\begin{equation*}
	\begin{pmatrix}
	1-\eta_{\ell}\sigma_{\ell}&\sigma_{1}-\eta_{\ell-1}\sigma_{\ell+1}&\cdots&\sigma_{\ell-1}-\sigma_{2\ell-1}\eta_{1}\\
	-\eta_{\ell}\sigma_{\ell-1}&1-\eta_{\ell-1}\sigma_{\ell}&\cdots&\sigma_{\ell-2}-\sigma_{2\ell-2}\eta_{1}\\
	\vdots&\vdots&\vdots&\vdots\\
	-\eta_{\ell}&-\eta_{\ell-1}\sigma_{1}&\cdots&-\sigma_{\ell-1}\eta_{1}\\
	0&-\eta_{\ell-1}&\cdots&-\sigma_{\ell-2}\eta_{1}\\
	\vdots&\vdots&\vdots&\vdots\\
	0&0&\cdots&-\eta_{1}
	\end{pmatrix}\cdot 
	\begin{pmatrix}
	x_{0}\\x_{1}\\ \vdots\\x_{\ell-1}
	\end{pmatrix}
	=\begin{pmatrix}
	d_{2\ell-1}\\ \vdots\\d_{\ell}\\d_{\ell-1}\\ \vdots\\  d_{0}
	\end{pmatrix},
	\end{equation*}	
	where $d_0,d_1,\cdots, d_{n-k+\ell-1}$ satisfy $ \prod\limits_{i\in J}(x-\alpha_{i})=\sum\limits_{i=0}^{n-k+\ell-1}d_{i}x^{n-k+\ell-1-i}$.
\end{corollary}

\section{The self-dual TGRS codes}
In this section, we study self-dual TGRS codes. Recall that an $[n, k]$ linear code $\mathcal{C}$ over $\mathbb{F}_{q}$ is called a self-dual code if $\mathcal{C}=\mathcal{C}^{\perp}$. If $\mathcal{C}$ has generator matrix $G$ and parity check matrix $H$, then $\mathcal{C}=\mathrm{span}_{\mathbb{F}_{q}}(G)$ and $\mathcal{C}^{\perp}=\mathrm{span}_{\mathbb{F}_{q}}(H)$. Therefore, $\mathcal{C}$ is self-dual if and only if $\mathrm{span}_{\mathbb{F}_{q}}(G)=\mathrm{span}_{\mathbb{F}_{q}}(H)$.

In the following, we always assume the TGRS code $\mathcal{C}=e v_{\boldsymbol{\alpha}, \boldsymbol{v}}(\mathcal{S})$ in $(\ref{equ:2.2})$ and $ n=2k$. Obviously, $\mathcal{C}$ has generator matrix
\begin{equation}\label{equ:5.1}
G=\begin{pmatrix}
v_{1}&\cdots&v_{n}\\
v_{1}\alpha_{1}&\cdots&v_{n}\alpha_{n}\\
\vdots&\vdots&\vdots\\
v_{1}\alpha_{1}^{k-\ell-1}&\cdots&v_{n}\alpha_{n}^{k-\ell-1}\\
v_{1}\left(\alpha_{1}^{k-\ell}+\eta_{1}\alpha_{1}^k\right)&\cdots&v_{n}\left(\alpha_{n}^{k-\ell}+\eta_{1}\alpha_{n}^k\right)\\
\vdots&\vdots&\vdots\\
v_{1}\left(\alpha_{1}^{k-1}+\eta_{\ell}\alpha_{1}^{k+\ell-1}\right)&\cdots&v_{n}\left(\alpha_{n}^{k-1}+\eta_{\ell}\alpha_{n}^{k+\ell-1}\right)\\
\end{pmatrix}
\end{equation}
and $\mathcal{C}$ has parity check matrix
\begin{equation}\label{equ:5.2}
H=\begin{pmatrix}
\cdots&\frac{u_{j}}{v_{j}}&\cdots\\
\cdots&\frac{u_{j}}{v_{j}}\alpha_{j}&\cdots\\
\vdots&\vdots&\vdots\\
\cdots&\frac{u_{j}}{v_{j}}\alpha_{j}^{n-k-\ell-1}&\cdots\\
\cdots&\frac{u_{j}}{v_{j}}\alpha_{j}^{n-k-\ell}\left(1-\eta_{\ell}\sum\limits_{i=0}^{\ell}\sigma_{\ell-i}\alpha_{j}^i\right)&\cdots\\
\cdots&\frac{u_{j}}{v_{j}}\alpha_{j}^{n-k-\ell}\left(\sum\limits_{i=0}^{1}\sigma_{1-i}\alpha_{j}^i-\eta_{\ell-1}\sum\limits_{i=0}^{\ell+1}\sigma_{\ell+1-i}\alpha_{j}^i\right)&\cdots\\
\vdots&\vdots&\vdots\\
\cdots&\frac{u_{j}}{v_{j}}\alpha_{j}^{n-k-\ell}\left(\sum\limits_{i=0}^{\ell-1}\sigma_{\ell-1-i}\alpha_{j}^i-\eta_{1}\sum\limits_{i=0}^{2\ell-1}\sigma_{2\ell-1-i}\alpha_{j}^i\right)&\cdots\\
\end{pmatrix}_{(n-k)\times n},
\end{equation}
where $\sigma_{t}\,(0 \leq t \leq 2\ell-1)$ is the $t$-th elementary symmetric polynomial of $\alpha_1,\alpha_2,\cdots, \alpha_n$, i.e., $\prod\limits_{i=1}^{n}\left(x-\alpha_{i}\right)=\sum\limits_{j=0}^{n} \sigma_{j} x^{n-j}$.

The theorem~\cite[Theorem 2.8]{huang2021mds} is important in determining the self-dualness of TGRS codes with a single twist. We generalize it as in the following lemma.

\begin{lemma}\label{lemma:representationGenMatParMat}
	Let $n=2k$ with $\ell\leq\lfloor\frac{k-1}{3}\rfloor$. Let $G$ in $(\ref{equ:5.1})$ and $H$ in $(\ref{equ:5.2})$ be the generator matrix and parity check matrix of $\mathcal{C}$, respectively. Let $\boldsymbol{g}_{i}$ and $\boldsymbol{h}_{i}$ denote the $(i+1)$-th row of $G$ and $H$, respectively. If $\eta_{1}\cdots\eta_{\ell}\neq 0$, then  $\left\{\boldsymbol{g}_{0},\boldsymbol{g}_{1}, \ldots, \boldsymbol{g}_{k-1}\right\}$ and $\left\{\boldsymbol{h}_{0},\boldsymbol{h}_{1}, \ldots, \boldsymbol{h}_{k-1}\right\}$ are linear representation of each other, if and only if the following condition hold:
	
	  $(1)$ $\left\{\boldsymbol{g}_{0}, \boldsymbol{g}_{1}, \ldots, \boldsymbol{g}_{k-\ell-1}\right\}$ and $\left\{\boldsymbol{h}_{0}, \boldsymbol{h}_{1}, \ldots, \boldsymbol{h}_{k-\ell-1}\right\}$ are linear representation of each other.
	
	  $(2)$ $\left\{\boldsymbol{g}_{k-\ell},\boldsymbol{g}_{k-\ell+1},\cdots, \boldsymbol{g}_{k-1}\right\}$ and $\left\{\boldsymbol{h}_{k-\ell},\boldsymbol{h}_{k-\ell+1},\cdots, \boldsymbol{h}_{k-1}\right\}$ are linear representation of each other.
\end{lemma}

\begin{proof}
	$\Leftarrow$\ It's obvious.\\
	$\Rightarrow$\ $(1)$\ Because $\ell\leq\lfloor\frac{k-1}{3}\rfloor$, $\forall i\in\{0,1,\cdots,k-\ell-1\}$, $\exists j\in\{0,1,\cdots,k-\ell-1\}$ such that $\left|i-j\right|=\ell$. For simplicity, we suppose that $0\leq i<j=i+\ell\leq k-\ell-1$. If $\left\{\boldsymbol{g}_{0},\boldsymbol{g}_{1},\cdots,\boldsymbol{g}_{k-1}\right\}$ and $\left\{\boldsymbol{h}_{0},\boldsymbol{h}_{1},\cdots,\boldsymbol{h}_{k-1}\right\}$ are representation of each other, then $\boldsymbol{g}_{i},\boldsymbol{g}_{j}\in \mathrm{span}_{\mathbb{F}_{q}}\{\boldsymbol{h}_{0},\boldsymbol{h}_{1},\cdots,\boldsymbol{h}_{k-1}\}$. In other words, $\boldsymbol{g}_{i}=(a_{0},a_{1},\cdots,a_{k-1})H,\boldsymbol{g}_{j}=(b_{0},b_{1},\cdots,b_{k-1})H$ with $a_{0},a_{1},\cdots,a_{k-1}$ not all zeros elements in $\mathbb{F}_{q}$ and $b_{0},b_{1},\cdots,b_{k-1}$ not all zero elements in $\mathbb{F}_{q}$, i.e., there exists $$f(x)=\sum\limits_{i=0}^{k-\ell-1}a_{i}x^i+\sum\limits_{i=0}^{\ell-1}a_{k-\ell+i}m_{i}(x),g(x)=\sum\limits_{i=0}^{k-\ell-1}b_{i}x^i+\sum\limits_{i=0}^{\ell-1}b_{k-\ell+i}m_{i}(x)$$ such that 
	\begin{equation*}
	\frac{v_{t}^2}{u_{t}}\alpha_{t}^i=f(\alpha_{t}),\quad \frac{v_{t}^2}{u_{t}}\alpha_{t}^j=g(\alpha_{t}),1\leq t\leq n,
	\end{equation*}
	where  $$m_{i}(x)=x^{k-\ell}(\sum\limits_{j=0}^i\sigma_{i-j}x^j-\eta_{\ell-i}\sum\limits_{j=0}^{i+\ell}\sigma_{i+\ell-j}x^j),0\leq i\leq \ell-1.$$ 
	So $\alpha_{t}^\ell f(\alpha_{t})=g(\alpha_{t}),t=1,2,\cdots,n$. Because $\alpha_{1},\cdots,\alpha_{n}$ are different roots of $f(x)x^\ell-g(x)$ and $\deg(f(x)x^\ell-g(x))\leq k+\ell-1+\ell<n$, we then obtain $f(x)x^\ell-g(x)=0$. Consequently, cofficients of $f(x)x^\ell-g(x)$ are equal to 0. So we have  
	\begin{equation*}
	C_{1}^T \begin{pmatrix}
	a_{k-\ell}\\a_{k-\ell+1}\\ \vdots\\a_{k-1}\end{pmatrix}=C_{2}^T\begin{pmatrix}
	b_{k-\ell}\\b_{k-\ell+1}\\ \vdots\\ b_{k-1}
	\end{pmatrix},C_{2}^T \begin{pmatrix}
	a_{k-\ell}\\a_{k-\ell+1}\\ \vdots\\a_{k-1}\end{pmatrix}
	=\begin{pmatrix}
	0\\0\\ \vdots\\0	
	\end{pmatrix},\end{equation*}
	where \begin{equation*}
	C_{1}=\left(\begin{array}{cccc}
	1-\eta_{\ell}\sigma_{\ell}&-\eta_{\ell}\sigma_{\ell-1}&\cdots&-\eta_{\ell}\sigma_{1}\\
	\sigma_{1}-\eta_{\ell-1}\sigma_{\ell+1}&1-\eta_{\ell-1}\sigma_{\ell}&\cdots&-\eta_{\ell-1}\sigma_{2}\\
	\vdots&\vdots&\ddots&\vdots\\
	\sigma_{\ell-1}-\eta_{1}\sigma_{2\ell-1}&\sigma_{\ell-2}-\eta_{1}\sigma_{2\ell-2}&\cdots&1-\eta_{1}\sigma_{\ell}
	\end{array}\right)
	\end{equation*}
	and \begin{equation*}
	C_{2}=\left(\begin{array}{ccccc}
	-\eta_{\ell}&0&0&\cdots&0\\
	-\eta_{\ell-1}\sigma_{1}&-\eta_{\ell-1}&0&\cdots&0\\
	\vdots&\vdots&\vdots&\ddots&\vdots\\
	-\eta_{1}\sigma_{\ell-1}&-\eta_{1}\sigma_{\ell-2}&-\eta_{1}\sigma_{\ell-3}&\cdots&-\eta_{1}
	\end{array}\right).
	\end{equation*}
	
	From the above linear equations, we can obtain $a_{k-\ell}=\cdots=a_{k-1}=b_{k-\ell}=\cdots=b_{k-1}=0$. In other words, $\boldsymbol{g}_{i},\boldsymbol{g}_{j}\in\mathrm{span}_{\mathbb{F}_{q}}\{\boldsymbol{h}_{0},\boldsymbol{h}_{1},\cdots,\boldsymbol{h}_{k-\ell-1}\}$. So $\mathrm{span}_{\mathbb{F}_{q}}\{\boldsymbol{g}_{0},\boldsymbol{g}_{1},\cdots,\boldsymbol{g}_{k-\ell-1}\}\subseteq\mathrm{span}_{\mathbb{F}_{q}}\{\boldsymbol{h}_{0},\boldsymbol{h}_{1},\cdots,\boldsymbol{h}_{k-\ell-1}\}$. It is obvious that $\dim(\mathrm{span}_{\mathbb{F}_{q}}\{\boldsymbol{g}_{0},\boldsymbol{g}_{1},\cdots,\boldsymbol{g}_{k-\ell-1}\})= \dim(\mathrm{span}_{\mathbb{F}_{q}}\{\boldsymbol{h}_{0},\boldsymbol{h}_{1},\cdots,\boldsymbol{h}_{k-\ell-1}\})=k-\ell$. Thus, $\left\{\boldsymbol{g}_{0},\boldsymbol{g}_{1},\cdots,\boldsymbol{g}_{k-\ell-1}\right\}$ and $\left\{\boldsymbol{h}_{0},\boldsymbol{h}_{1},\cdots,\boldsymbol{h}_{k-\ell-1}\right\}$ are linear representation of each other.
	
	(2)\ For each $ j\in\{k-\ell,\cdots,k-1\}$, due to $\mathrm{span}_{\mathbb{F}_{q}}\{\boldsymbol{g}_{0},\boldsymbol{g}_{1},\cdots,\boldsymbol{g}_{k-\ell-1}\}=\mathrm{span}_{\mathbb{F}_{q}}\{\boldsymbol{h}_{0},\boldsymbol{h}_{1},\cdots,\boldsymbol{h}_{k-\ell-1}\}$, thus $\boldsymbol{g}_{0}=(c_{0},c_{1},\cdots,c_{k-\ell-1})(\boldsymbol{h}_{0}^T,\boldsymbol{h}_{1}^T,\cdots,\boldsymbol{h}_{k-\ell-1}^T)^T$ with $c_{0},c_{1},\cdots,c_{k-\ell-1}$ not all zero elements in $\mathbb{F}_{q}$. That is, there exists $h(x)=\sum\limits_{i=0}^{k-\ell-1}c_{i}x^i\in \mathbb{F}_{q}[x]$ such that \begin{equation*}
	\frac{v_{t}^2}{u_{t}}=h(\alpha_{t}),\quad 1\leq t\leq n.
	\end{equation*}
	Moreover, $\boldsymbol{g}_{j}=(d_{0},d_{1},\cdots,d_{k-1})H$ with $d_{0},d_{1},\cdots,d_{k-1}$ not all zero elements in $\mathbb{F}_{q}$. That is, there exists $p(x)=\sum\limits_{i=0}^{k-\ell-1}d_{i}x^i+\sum\limits_{i=0}^{\ell-1}d_{k-\ell+i}m_{i}(x)\in \mathbb{F}_{q}[x]$ such that \begin{equation*}
	\frac{v_{t}^2}{u_{t}}\left(\alpha_{t}^j+\eta_{j-k+\ell+1}\alpha_{t}^{j+\ell}\right)=p(\alpha_{t}),1\leq t\leq n.
	\end{equation*}
	Noting that $$\deg(h(x)\left(x^j+\eta_{j-k+\ell+1}x^{j+\ell}\right)-p(x))\leq n-2<n$$ and $\alpha_{1},\cdots,\alpha_{n}$ are different roots of  $$h(x)\left(x^j+\eta_{j-k+\ell+1}x^{j+\ell}\right)-p(x),$$ 
	we then obtain $$h(x)(x^j+\eta_{j-k+\ell+1}x^{j+\ell})=p(x).$$ Consequently, coefficients of $h(x)(x^j+\eta_{j-k+\ell+1}x^{j+\ell})-p(x)$ are equal to $0$. We then obtain $$d_{0}=d_{1}=\cdots=d_{k-\ell-1}=0.$$ In other words, $\boldsymbol{g}_{j}\in\mathrm{span}_{\mathbb{F}_{q}}\{\boldsymbol{h}_{k-\ell},\boldsymbol{h}_{k-\ell+1},\cdots,\boldsymbol{h}_{k-1}\}$. Thus, 
	$$\mathrm{span}_{\mathbb{F}_{q}}\{\boldsymbol{g}_{k-\ell},\cdots,\boldsymbol{g}_{k-1}\}\subseteq \mathrm{span}_{\mathbb{F}_{q}}\{\boldsymbol{h}_{k-\ell},\cdots,\boldsymbol{h}_{k-1}\}.$$ On the other hand,  $\dim(\mathrm{span}_{\mathbb{F}_{q}}\{\boldsymbol{g}_{k-\ell},\cdots,\boldsymbol{g}_{k-1}\})=\dim( \mathrm{span}_{\mathbb{F}_{q}}\{\boldsymbol{h}_{k-\ell},\cdots,\boldsymbol{h}_{k-1}\})$. Thus, $\left\{\boldsymbol{g}_{k-\ell},\cdots,\boldsymbol{g}_{k-1}\right\}$ and $\left\{\boldsymbol{h}_{k-\ell},\cdots,\boldsymbol{h}_{k-1}\right\}$ are linear representation of each other.
\end{proof}

\begin{theorem}\label{thmcont:5.2}
	Let $n=2k$ with $\ell\leq \lfloor\frac{k-1}{3}\rfloor$. Let $\alpha_{1},\alpha_{2},\cdots,\alpha_{n}$ be distinct elements of $\mathbb{F}_{q}$, $\prod\limits_{i=1}^{n}(x-\alpha_{i})=\sum\limits_{j=0}^{n}\sigma_{j}x^{n-j}$
	and $u_{i}=\prod\limits_{j=1,j\neq i}^{n}(\alpha_{i}-\alpha_{j})^{-1}$
	for $1\leq i\leq n$. Let $v_{i}\in \mathbb{F}_{q}^{*}$ 	for $1\leq i\leq n$ and $\prod\limits_{i=1}^{\ell}\eta_{i}\neq 0$. Then
	$\mathcal{C}= ev_{\alpha,v}(\mathcal{S})$ is self-dual if and only if the following conditions hold:
	
	$(1)$\ There exists a $\lambda\in \mathbb{F}_{q}^{*}$ such that $v_{i}^2=\lambda u_{i}$ for all $1\leq i\leq n$.
	
	$(2)$\ $\sigma_{1}=\cdots=\sigma_{\ell-1}=\sigma_{\ell+1}=\cdots=\sigma_{2\ell-1}=0$ and $\frac{1}{\eta_{i}}+\frac{1}{\eta_{\ell+1-i}}=\sigma_{\ell},i=1,2,\cdots,\lceil\frac{\ell+1}{2}\rceil$.
\end{theorem}
\begin{proof}
	We know that $\mathcal{C}$ has generator matrix $G$ as $(\ref{equ:5.1})$ and parity check matrix $H$ as $(\ref{equ:5.2})$. Let $\boldsymbol{g}_{i}$ and $\boldsymbol{h}_{i}$ denote the $(i+1)$-th row of $G$ and $H$, respectively. By  Lemma~\ref{lemma:representationGenMatParMat},  $\mathcal{C}$ is self-dual if and only if $\left\{\boldsymbol{g}_{0}, \ldots, \boldsymbol{g}_{k-1}\right\}$ and $\left\{\boldsymbol{h}_{0}, \ldots, \boldsymbol{h}_{k-1}\right\}$ are linear representation of each other, if and only if (1) $\left\{\boldsymbol{g}_{0},  \ldots, \boldsymbol{g}_{k-\ell-1}\right\}$ and $\left\{\boldsymbol{h}_{0},  \ldots, \boldsymbol{h}_{k-\ell-1}\right\}$ are linear representation of each other and (2) $\left\{\boldsymbol{g}_{k-\ell},\cdots, \boldsymbol{g}_{k-1}\right\}$ and $\left\{\boldsymbol{h}_{k-\ell},\cdots, \boldsymbol{h}_{k-1}\right\}$ are linear representation of each other.
	
	 Let $\boldsymbol{\alpha}^{i}=\left(\alpha_{1}^{i}, \alpha_{2}^{i}, \ldots, \alpha_{n}^{i}\right)$ and $\boldsymbol{\frac{u}{v}}=\left(\frac{u_{1}}{v_{1}}, \frac{u_{2}}{v_{2}}, \ldots, \frac{u_{n}}{v_{n}}\right)$. Similar to the proof of \cite[Theorem 2.8]{huang2021mds}, we know that  $\left\{\boldsymbol{g}_{0}, \ldots, \boldsymbol{g}_{k-\ell-1}\right\}$ and $\left\{\boldsymbol{h}_{0}, \ldots, \boldsymbol{h}_{k-\ell-1}\right\}$ are linear representation of each other if and only if $\boldsymbol{v}=\lambda \boldsymbol{\frac{u}{v}}$, for some $\lambda \in \mathbb{F}_{q}^{*}$. On the other hand,
	\begin{equation*}
	\begin{aligned}
	&{\left(\begin{array}{c}
	\boldsymbol{v} * \boldsymbol{\alpha}^{k-\ell} \\
	\boldsymbol{v} * \boldsymbol{\alpha}^{k-\ell+1} \\
	\vdots\\
	\boldsymbol{v} * \boldsymbol{\alpha}^{k+\ell-1}   \\
	\end{array}\right)}^T
	{\left(\begin{array}{cccccccc}
	1 & 0 & \cdots&0 & \eta_{1} & 0 &\cdots& 0 \\
	0 & 1 & \cdots&0 &   0 &\eta_{2}&\cdots & 0\\
	\vdots&\vdots&\ddots&\vdots&\vdots&\vdots&\ddots&\vdots\\
	0&0&\cdots&1&0&0&\cdots&\eta_{\ell}
	\end{array}\right)}^T={\left(\begin{array}{c}
	\boldsymbol{g}_{k-\ell}\\
	\boldsymbol{g}_{k-\ell+1} \\
	\vdots\\
	\boldsymbol{g}_{k-1}
	\end{array}\right)}^T,
	 \\
	&{\left(\begin{array}{c}
		\boldsymbol{\frac{u}{v}} * \boldsymbol{\alpha}^{k-\ell} \\
		\boldsymbol{\frac{u}{v}} * \boldsymbol{\alpha}^{k-\ell+1} \\
		\vdots\\
		\boldsymbol{\frac{u}{v}} * \boldsymbol{\alpha}^{k+\ell-1}   
		\end{array}\right)}^{T}\begin{pmatrix}
	1-\eta_{\ell}\sigma_{\ell}&\sigma_{1}-\eta_{\ell-1}\sigma_{\ell+1}&\cdots&\sigma_{\ell-1}-\sigma_{2\ell-1}\eta_{1}\\
	-\eta_{\ell}\sigma_{\ell-1}&1-\eta_{\ell-1}\sigma_{\ell}&\cdots&\sigma_{\ell-2}-\sigma_{2\ell-2}\eta_{1}\\
	\vdots&\vdots&\vdots&\vdots\\
	-\eta_{\ell}&-\eta_{\ell-1}\sigma_{1}&\cdots&-\sigma_{\ell-1}\eta_{1}\\
	0&-\eta_{\ell-1}&\cdots&-\sigma_{\ell-2}\eta_{1}\\
	\vdots&\vdots&\vdots&\vdots\\
	0&0&\cdots&-\eta_{1}
	\end{pmatrix}={\left(\begin{array}{c}
		\boldsymbol{h}_{k-\ell} \\
		\boldsymbol{h}_{k-\ell+1}\\
		\vdots\\
		\boldsymbol{h}_{k-1}
		\end{array}\right)}^T\\
	\end{aligned}
	\end{equation*}
	where $*$ denotes componentwise product. Then $\left\{\boldsymbol{g}_{k-\ell},\boldsymbol{g}_{k-\ell+1},\cdots,\boldsymbol{g}_{k-1}\right\}$ and $\left\{\boldsymbol{h}_{k-\ell},\boldsymbol{h}_{k-\ell+1},\cdots,\boldsymbol{h}_{k-1}\right\}$ are linear representation of each other if and only if the $\rank$ of these two coefficient matrices is equal and
	\begin{equation}\label{equ:5.3}
	\begin{pmatrix}
	1-\eta_{\ell}\sigma_{\ell}\\-\eta_{\ell}\sigma_{\ell-1}\\
	\vdots\\-\eta_{\ell}\\0 \\ \vdots\\0
	\end{pmatrix}
	,\begin{pmatrix}
	\sigma_{1}-\eta_{\ell-1}\sigma_{\ell+1}\\
	1-\eta_{\ell-1}\sigma_{\ell}\\
	\vdots\\ -\eta_{\ell-1}\sigma_{1}\\
	-\eta_{\ell-1}\\
	\vdots\\0
	\end{pmatrix},\cdots,
	\begin{pmatrix}
	\sigma_{\ell-1}-\sigma_{2\ell-1}\eta_{1}\\
	\sigma_{\ell-2}-\sigma_{2\ell-2}\eta_{1}\\
	\vdots\\-\sigma_{\ell-1}\eta_{1}\\-\sigma_{\ell-2}\eta_{1}\\ \vdots\\ -\eta_{1}
	\end{pmatrix}
	\end{equation} 
	can be linearly expressed by 
	\begin{equation}\label{equ:5.4}
	\begin{pmatrix}
	1\\0\\ \vdots\\0\\ \eta_{1}\\0 \\ \vdots\\0
	\end{pmatrix},\begin{pmatrix}
	0\\1\\ \vdots\\0\\ 0\\ \eta_{2}\\ \vdots\\0
	\end{pmatrix},\cdots
	,\begin{pmatrix}
	0\\0\\ \vdots\\1\\ 0\\0\\ \vdots\\ \eta_{\ell}
	\end{pmatrix}.
	\end{equation}
	Obviously, the $\rank$ of these two coefficient matrices is equal, so $\left\{\boldsymbol{g}_{k-\ell},\cdots,\right.$ $\left.\boldsymbol{g}_{k-1}\right\}$ and $\left\{\boldsymbol{h}_{k-\ell},\cdots,\boldsymbol{h}_{k-1}\right\}$ are linear representation of each other, if and only if $(\ref{equ:5.3})$ can be linearly expressed by $(\ref{equ:5.4})$, if and only if \begin{equation*}
	\left\{
	\begin{array}{cc}
	\eta_{1}(1-\eta_{\ell}\sigma_{\ell})=-\eta_{\ell}\\
	-\eta_{2}\eta_{\ell}\sigma_{\ell-1}=0\\
	\vdots\\
	-\eta_{\ell}^2\sigma_{1}=0
	\end{array}\right.,
	\left\{
	\begin{array}{cc}
	\eta_{1}(\sigma_{1}-\eta_{\ell-1}\sigma_{\ell+1})=-\eta_{\ell-1}\sigma_{1}\\
	\eta_{2}(1-\eta_{\ell-1}\sigma_{\ell})=-\eta_{\ell-1}\\
	\vdots\\
	-\eta_{\ell-1}\eta_{\ell}\sigma_{2}=0
	\end{array}\right.,\cdots,
	\end{equation*}
	\begin{equation*}
	\left\{
	\begin{array}{ll}
	\eta_{1}(\sigma_{\ell-1}-\sigma_{2\ell-1}\eta_{1})=-\sigma_{\ell-1}\eta_{1}\\
	\eta_{2}(\sigma_{\ell-2}-\sigma_{2\ell-2}\eta_{1})=-\sigma_{\ell-2}\eta_{1}\\
	\quad\quad\quad\quad\quad\vdots\\
	\eta_{\ell}(1-\sigma_{\ell}\eta_{1})=-\eta_{1}
	\end{array}\right.,
	\end{equation*}
	if and only if $\sigma_{1}=\sigma_{2}=\cdots=\sigma_{\ell-1}=\sigma_{\ell+1}=\sigma_{\ell+2}=\cdots=\sigma_{2\ell-1}=0$ and $\frac{1}{\eta_{i}}+\frac{1}{\eta_{\ell+1-i}}=\sigma_{\ell},i=1,2,\cdots,\lceil\frac{\ell+1}{2}\rceil$. It completes the proof.
\end{proof}

For $\ell=2$, comparing with the result of~\cite{9828478}, the twists are different. And we obtain a new necessary and sufficient condition of $\mathcal{C}=ev_{\boldsymbol{\alpha},\boldsymbol{v}}(\mathcal{S})$ to be self-dual.
\begin{corollary}\label{corcont:5.3}
	Let $n=2k$ with $k\geq 6$. Let $\alpha_{1},\alpha_{2},\cdots,\alpha_{n}$ be distinct elements of $\mathbb{F}_{q}$, $\prod\limits_{i=1}^{n}(x-\alpha_{i})=\sum\limits_{j=0}^{n}\sigma_{j}x^{n-j}$
	and $u_{i}=\prod\limits_{j=1,j\neq i}^{n}(\alpha_{i}-\alpha_{j})^{-1}$
	for $1\leq i\leq n$. Let $v_{i}\in \mathbb{F}_{q}^{*}$ 	for $1\leq i\leq n$ and $\eta_{1},\eta_{2}\neq 0$. Then
	$\mathcal{C}= ev_{\boldsymbol{\alpha},\boldsymbol{v}}(\mathcal{S})$ is self-dual if and only if the following conditions hold:

    $(1)$\ There exists a $\lambda\in \mathbb{F}_{q}^{*}$ such that $v_{i}^2=\lambda u_{i}$ for all $1\leq i\leq n$.
    
     $(2)$\ $\sigma_{1}=0,\eta_{1}+\eta_{2}=\eta_{1}\eta_{2}\sigma_{2}$. 
\end{corollary}

Finally, we give an explicit construction of self-dual TGRS codes.
\begin{theorem}\label{thmcont:5.4}
	Let $q$ be an odd prime power such that $(q,\ell)=1$ and $7\ell\leq q^\ell-1$. Let $\mathbb{F}_{{q}^s}$ is the splitting field of $f(x)=x^\ell-a$ over $\mathbb{F}_{q}$, where $a\in \mathbb{F}_{q}^{*}$ and $s\leq\ell$.
	
	 $(1)$\ If $\ell$ is odd,  let  $m(x)=\frac{x^{q^s}-x}{f(x)}$ and  $\alpha_{i},1\leq i\leq q^s-\ell$ be all the roots of $m(x)$. There exist $v_{i}\in \mathbb{F}_{q^{2s}}$ such that $v_{i}^2=m^{'}(\alpha_{i})^{-1},1\leq i\leq q^s-\ell$.  Let $\eta_{i}\neq 0,a^{-1}$ and $\eta_{\ell+1-i}=\frac{1}{a-\eta_{i}^{-1}},i=1,\cdots,\frac{\ell+1}{2}$. Then $\mathcal{C}=ev_{\boldsymbol{\alpha},\boldsymbol{v}}(\mathcal{S})$ is a $[q^s-\ell,\frac{q^s-\ell}{2},\geq \frac{q^s-3\ell+2}{2}]$ self-dual code over $\mathbb{F}_{q^{2s}}$.
	 
	  $(2)$\ If $\ell$ is even, let  $m(x)=\frac{x^{q^s}-x}{xf(x)}$ and  $\alpha_{i},1\leq i\leq q^s-\ell-1$ be all the roots of $m(x)$. There exist $v_{i}\in \mathbb{F}_{q^{2s}}$ such that $v_{i}^2=m^{'}(\alpha_{i})^{-1},1\leq i\leq q^s-\ell-1$.  Let $\eta_{i}\neq 0,a^{-1}$ and $\eta_{\ell+1-i}=\frac{1}{a-\eta_{i}^{-1}},i=1,\cdots,\frac{\ell}{2}$. Then $\mathcal{C}=ev_{\boldsymbol{\alpha},\boldsymbol{v}}(\mathcal{S})$ is a $[q^s-\ell-1,\frac{q^s-\ell-1}{2},\geq \frac{q^s-3\ell+1}{2}]$ self-dual code over $\mathbb{F}_{q^{2s}}$.
\end{theorem}
\begin{proof}
$(1)$ If $\ell$ is odd, since $f(x)$ has $\ell$  roots in $\mathbb{F}_{q^s}$ and $(x^\ell-a,\ell x^{\ell-1})=1$, $f(x)$ has $\ell$ distinct roots in $\mathbb{F}_{q^s}$. Thus, $m(x)=\frac{x^{q^s}-x}{f(x)}$ has $q^s-\ell$ distinct roots in $\mathbb{F}_{q^s}$. Note that $m^{'}(\alpha_{i})\in \mathbb{F}_{q^s}$ has square roots in $\mathbb{F}_{q^{2s}}$. So there exist $v_{i}\in \mathbb{F}_{q^{2s}}$ such that $v_{i}^2=m^{'}(\alpha_{i})^{-1}=u_{i}$. Write $m(x)=\sum\limits_{i=0}^{q^s-\ell}m_{i}x^{q^s-\ell-i}$. Since $x^{q^s}-x=f(x)m(x)$, we have $m_{0}=1,m_{1}=\cdots=m_{\ell-1}=0,m_{\ell}=a\cdot m_{0}=a,m_{\ell+1}=am_{1}=0,\cdots,m_{2\ell-1}=am_{\ell-1}=0$. On the other hand, from the constructions of $\eta_{1},\cdots,\eta_{\ell}\in \mathbb{F}_{q}^{*}$, it is easy to see that they satisfy $\frac{1}{\eta_{i}}+\frac{1}{\eta_{\ell+1-i}}=a,i=1,\cdots,\frac{\ell+1}{2}$. Therefore, by Theorem~\ref{thmcont:5.2} and $3\ell\leq \frac{q^\ell-\ell}{2}=k$,  $\mathcal{C}=ev_{\boldsymbol{\alpha},\boldsymbol{v}}(\mathcal{S})$ is a self-dual code of length $q^s-\ell$ over $\mathbb{F}_{q^{2s}}$. Furthermore, it is obvious that $d(\mathcal{C})\geq\frac{q^s-3\ell+2}{2}$.
	
$(2)$ If $\ell$ is even, by the same argument, we can easily prove that $\mathcal{C}=ev_{\boldsymbol{\alpha},\boldsymbol{v}}(\mathcal{S})$ is a $[q^s-\ell-1,\frac{q^s-\ell-1}{2},\geq \frac{q^s-3\ell+1}{2}]$ self-dual code over $\mathbb{F}_{q^{2s}}$.
\end{proof}
 
 \begin{example}
 	$(1)$\ Let $q=13,\ell=3$, $\boldsymbol{\alpha}=\left(0,1,2,3,4,5,6,9,10,12\right)$ and $f(x)=x^3-5$. Since the polynomial $f(x)$ factors as $(x+2)(x-7)(x-8)$ in $\mathbb{F}_{13}$, the splitting field of $f(x)$ over $\mathbb{F}_{13}$ is still $\mathbb{F}_{13}$. It is easy to compute that $m(x)=\frac{x^{13}-x}{f(x)}=x^{10}+ 5x^7 + 12x^4 + 8x$. 
 	Let $\mathbb{F}_{13^2}^*=\langle\beta\rangle$, where the minimal polynomial of $\beta$ over $\mathbb{F}_{13}$ is $x^2+7x+2$. Choose $\boldsymbol{v}=\left(\beta^{63},2,6,2,\beta^{35},6,6,2,\beta^{35},\beta^{35}\right)$, then $v_{i}^2=m^{'}(\alpha_{i})^{-1},1\leq i\leq 10$. Let $\eta_{1}=2,\eta_{2}=3,\eta_{3}=6$ and let \begin{equation*}
 	\mathcal{S}=\left\{\sum\limits_{i=0}^{4}f_{i}x^i+2f_{2}x^5+3f_{3}x^6+6f_{4}x^7:\mbox{for all}\  f_{i}\in\mathbb{F}_{13^2},0\leq i\leq 4\right\}.
 	\end{equation*}
 	By Theorem~\ref{thmcont:5.4}, $\mathcal{C}=ev_{\boldsymbol{\alpha},\boldsymbol{v}}(\mathcal{S})$ is self-dual. Together with Example~\ref{exacont:3.7}, the TGRS code $\mathcal{C}=ev_{\boldsymbol{\alpha},\boldsymbol{v}}(\mathcal{S})$ is a self-dual MDS code with parameters $[10,5,6]$.
 	
 	$(2)$\ Let $q=13,\ell=4$, $\boldsymbol{\alpha}=\left(1,4,5,6,7,8,9,12\right)$, and $f(x)=x^4-3$. Since the polynomial $f(x)$ factors as $(x-2)(x-3)(x-10)(x-11)$ in $\mathbb{F}_{13}$, the splitting field of $f(x)$ over $\mathbb{F}_{13}$ is still $\mathbb{F}_{13}$. One can easily show that $m(x)=\frac{x^{13}-x}{xf(x)}=x^8+3x^4+9$. Since $v_{i}^2=m^{'}(\alpha_{i})^{-1}$, then we have $\boldsymbol{v}^2=\left(2,2,10,3,10,3,11,11\right)$.
 	Let $\mathbb{F}_{13^2}^*=\langle\beta\rangle$, where the minimal polynomial of $\beta$ is $x^2+7x+2$. Choose $\boldsymbol{v}=\left(\beta^7,\beta^{7},6,4,6,4,\beta^{49},\beta^{49}\right)$, then $v_{i}^2=m^{'}(\alpha_{i})^{-1},1\leq i\leq 8$. Let $\eta_{1}=1,\eta_{2}=3,\eta_{3}=2,\eta_{4}=7$ and let \begin{equation*}
 	\mathcal{S}=\left\{\sum\limits_{i=0}^{3}f_{i}\left(x^i+\eta_{i+1}x^{4+i}\right)
 	:\mbox{for all}\  f_{i}\in\mathbb{F}_{13^2},0\leq i\leq 3\right\}.
 	\end{equation*}
 	By Theorem~\ref{thmcont:5.4}, the TGRS code $\mathcal{C}=ev_{\boldsymbol{\alpha},\boldsymbol{v}}(\mathcal{S})$ is self-dual. Furthermore, the TGRS code $\mathcal{C}=ev_{\boldsymbol{\alpha},\boldsymbol{v}}(\mathcal{S})$ is indeed a self-dual MDS code with parameters $[8,4,5]$.
 \end{example}

\section{Conclusion}\label{sect:6}
In this paper, we have characterized a sufficient and necessary condition that a TGRS code with $\ell$ twists is MDS, AMDS, NMDS or $\ell$-MDS for $\ell\leq \min\{k,n-k\}$. Also, we have determined a sufficient and necessary condition that a TGRS code with $\ell$ twists is self-dual for $\ell\leq \lfloor \frac{k-1}{3}\rfloor$, and given an explicit construction of self-dual TGRS code.

\bibliographystyle{plain}
\bibliography{references-sMDS}

\end{document}